\documentclass[10pt]{article}

\usepackage{amsmath}
\usepackage{amssymb}

\usepackage{graphicx}

\usepackage{cite}
\usepackage{color}
\usepackage{soul}

\usepackage[labelfont=bf,labelsep=period,justification=raggedright]{caption}

\makeatletter
\renewcommand{\@biblabel}[1]{\quad#1.}
\makeatother

\date{}

\usepackage{amsthm}

\newtheorem{lemma}{Lemma}
\newtheorem{theorem}{Theorem}

\begin{document}

\begin{flushleft}
{\Large
\textbf{The Stubborn Roots of Metabolic Cycles}
}
\\
Ed Reznik$^{1,\ast}$, 
Alex Watson$^{2}$, 
Osman Chaudhary$^{3}$
\\
\bf{1} Department of Biomedical Engineering, Boston University, Boston MA 02215
\\
\bf{2} Department of Applied Physics and Applied Mathematics, Columbia University, New York NY 10027
\\
\bf{3} Department of Mathematics, Boston University, Boston MA 02215

$\ast$ E-mail: ereznik@bu.edu
\end{flushleft}

\section*{Summary}
Efforts to catalogue the structure of metabolic networks have generated highly detailed, genome-scale atlases of biochemical reactions in the cell. Unfortunately, these atlases fall short of capturing the kinetic details of metabolic reactions, instead offering only \textit{topological} information from which to make predictions. As a result, studies frequently consider the extent to which the topological structure of a metabolic network determines its dynamic behavior, irrespective of kinetic details. Here, we study a class of metabolic networks known as non-autocatalytic metabolic cycles, and analytically prove an open conjecture regarding the stability of their steady-states. Importantly, our results are invariant to the choice of kinetic parameters, rate laws, equilibrium fluxes, and metabolite concentrations. Unexpectedly, our proof exposes an elementary but apparently open problem of locating the roots of a sum of two polynomials $S = P+Q$, when the roots of the summand polynomials $P$ and $Q$ are known. We derive two new results named the Stubborn Roots Theorems, which provide sufficient conditions under which the roots of $S$ remain qualitatively identical to the roots of $P$. Our work illustrates how complementary feedback, from classical fields like dynamical systems to biology and vice versa, can expose fundamental and potentially overlooked questions.

\section*{Keywords}
Metabolic networks, Stability, Dynamical Systems, Polynomial, Roots

\section*{Introduction}
Networks of enzyme-catalyzed metabolic reactions are fundamental to the proliferation of life, using the energy extracted from environmental nutrients to drive the assembly of organic macromolecules and enable the successful reproduction of the cell. The large-scale architecture of these networks is rich in structure: they are broadly organized into overlapping pathways (\textit{e.g.} catabolic glycolysis and anabolic glucoeneogenesis), and exhibit power-law-like degree distributions with highly connected cofactors (\textit{e.g.} ATP/ADP and NADH/NAD) linking many otherwise-distant metabolites \cite{Ravasz2002}. Perhaps most importantly, these networks are capable of robust operation in spite of heterogeneity in the abundances of crucial enzymes and substrates \cite{Shinar2010}. 

To what extent are the robust features of metabolic networks determined by the underlying topological structure of the network itself? This question lies at the center of many studies precisely because contemporary metabolic models are largely limited to structural information. Using genomic data, it is now possible to reconstruct genome-scale models of metabolism, which predict the presence of absence or enzymes (and by virtue, metabolic reactions) in an organism. However, the kinetic details of these reactions (such as the rate laws they obey, as well as rate constants like $K_M$ or $V_{max}$) are largely unknown due to the difficulty of measuring them in a high-throughput manner \textit{in vivo}. As a result, a number of generic results linking the qualitative dynamics of a chemical reaction network (such as its potential for multistability or sustained oscillations) to its structural organization \cite{Feinberg1987, Feinberg1988} have begun to populate the literature, pointing to a fundamental connection between structure and dynamics. These results make remarkably minimal and physically reasonable assumptions on the generic form of kinetic rate laws, that render their conclusions largely independent of the choice of parameters. Perhaps the most well-known result from studies of this sort is the Deficiency Zero Theorem \cite{Feinberg1987} from chemical reaction network theory (CRNT), which gives sufficient conditions for unique equilibria and asymptotic stability of a large class of reaction networks. More recent work extending CRNT, such as that of Shinar and Feinberg \cite{Shinar2010}, has identified structural properties endowing chemical reaction networks with absolute concentration robustness (ACR) (\textit{i.e.} the steady-state concentration of a molecular species is identical in any steady-state the dynamical system admits). Importantly, other existing methodologies, including Chemical Organization Theory \cite{Dittrich2007, Kreyssig2012} and the theory of monotone systems \cite{Angeli2004}, take similar ``topological'' approaches to understanding how dynamics may be inferred from, and in fact directly influenced by, the structure of reaction networks themselves.

In prior work, we studied a family of non-autocatalytic metabolic cycles \cite{Reznik2010}, and considered the role that their cyclic topology might play in determining their steady-state properties. Our decision to study a cycle, as opposed to any other structure, was motivated by the prominent role of cycles (such as the TCA and Calvin cycles) in present-day metabolic networks \cite{Berg2002}. The approach we took relied on a method known as structural kinetic modeling (SKM), which applied a change of variables to the dynamical system corresponding to the metabolic cycle. This change of variables enabled us to study the dynamics of the metabolic cycle from a structural point of view, with only mild assumptions on the form of the kinetics themselves.

The main outcome of our work in \cite{Reznik2010} was \textit{limited numerical evidence} that any steady-state of the cycle must be stable to small perturbations, irrespective of equilibrium metabolite concentrations, flux magnitudes, or the choice of kinetic parameters. However, we were unable to offer a rigorous analytical proof of this claim. In particular, it was unclear whether small regions of parameter space harboring unstable equilibria might exist. Perhaps more importantly, computational considerations limited our numerical investigations to relatively short metabolic cycles (including up to eight metabolites), leaving open the possibility that instability appeared as cycles grew longer. As a result, we left the question of stability as an unproven conjecture (herein referred to as the \textit{cycle stability conjecture}). This conjecture is the object of study in the first part of this work.

The difficulty with proving the cycle stability conjecture reduced to locating the roots of a high-order polynomial in the complex plane. Although a number of classical results from control theory are commonly applied to problems like this, they were rendered largely unusable for the polynomial in \cite{Reznik2010}. For example, the Routh-Hurwitz (RH) criterion, perhaps the best-known technique for constraining the locations of a polynomial's roots, requires the precise calculation of coefficients of the polynomial under study. Unfortunately, the calculation of these coefficients for the polynomial in \cite{Reznik2010} became analytically intractable as the number of reactions in the cycle (which we assumed to be arbitrarily large) grew. Furthermore, because these coefficients were themselves functions of SKM variables, and were only constrained to lie in complicated intervals on the real axis, the problem of proving stability became substantially more difficult. Well-known methods, such as Kharitonov's theorems \cite{Dasgupta1988}, exist to study how uncertainty in the coefficients of a polynomial impacts its roots. In addition, the Method of Resultants, used by Gross and Feudel in the context of studying generalized models \cite{Gross2004}, can be used to identify when pairs of imaginary roots cross the imaginary axis (suggesting the onset of instability and sustained oscillations).  However, the complexity and scale of the polynomial in \cite{Reznik2010} again rendered the application of both Kharitonov's theorems and the Method of Resultants infeasible.

Our difficulty with bringing classical tools to bear on the cycle stability conjecture motivated us to revisit the problem from a completely new perspective. In this work, we resolve the cycle stability conjecture by reformulating it as a question of locating the roots of a sum of two polynomials. By doing so, we are able to apply a classical technique (Rouch\'{e}'s Theorem \cite{BrownJamesWardChurchill1996}) from complex analysis to resolve the conjecture and prove the stability of non-autocatalytic metabolic cycles. Quite unexpectedly, our proof leads us to a substantially more general question: how do the roots of a sum of polynomials $S = P+Q$ depend on the roots of $P$ and $Q$ themselves? Using a method identical to the one used to prove the cycle stability conjecture, we prove two new theorems, which we call the Stubborn Roots Theorems, which give sufficient conditions for when the roots of $S$ are qualitatively identical to the roots of $P$. To our knowledge, there are few generic results which provide information regarding the locations of the roots of a sum of two polynomials. Given the fundamental importance of locating the roots of a polynomial in the study of dynamical systems (and in applications of dynamical systems to fields such as in systems biology), we feel that the Stubborn Roots Theorems may find use in other contexts.

\section*{Results}

\subsection*{Stability of Metabolic Cycles}
We begin by presenting a model of the dynamics of a simple non-autocatalytic metabolic cycle, and then proceed to using SKM to study the stability of its equilibria. First, we describe the generic structure of the metabolic cycle under study, which is identical to the one studied in \cite{Reznik2010}. The cycle contains $n$ metabolites ($M_1 \hdots M_n$) and two cofactors ($O_1$ and $O_2$) which provide the energetic force to thermodynamically drive the metabolic reactions, and can be illustrated by

\begin{align}
\emptyset &\longrightarrow M_{1} \nonumber \\
M_{1} + O_{1} &\longrightarrow M_{2} + O_{2} \nonumber \\
M_{i} &\longrightarrow M_{i+1}, \hspace{2mm}  i = 2\hdots n-1 \nonumber \\
M_{n} &\longrightarrow M_{1} \nonumber \\
M_{n} &\longrightarrow \emptyset \nonumber \\
O_2 &\xrightarrow{energy} O_1 
\label{eq:reactions}
\end{align}

In this cycle, each metabolite $M_i$ is converted to metabolite $M_{i+1}$ for $i = 1 \hdots n-1$. A constant flux of $M_1$ enters the system. A proportion of the last metabolite $M_n$ is converted back to $M_1$, while the remainder leaves the system. The high-energy cofactor $O_1$ is converted to its low-energy cofactor partner $O_2$ in the reaction catalyzing the conversion of $M_1$ to $M_2$. In a separate reaction, energy is input into the system to drive the reformation of the higher energy molecule $O_1$. 

At steady-state, the magnitude of the flux through each reaction in the cycle can be calculated by enforcing mass balances on each metabolite. To do so, we assume that a constant flux of generic magnitude $\alpha v, 0<\alpha<1$ of metabolite $M_1$ flows into the network. A proportion $(1-\alpha)v$ of the flux entering $M_n$ is channeled back towards $M_1$, while the remaining flux $\alpha v$ exits the system.  All other reactions carry a steady-state flux of $v$. It is easily verified that this flux vector is in the nullspace of the stoichiometric matrix $\mathbf{S}$ (see Appendix). We assume that the kinetics of each reaction are monotonic, that is, that an increase in the concentration of substrate for any reaction will consequently increase the rate of the reaction. This assumption is quite generic, and is amenable with many well-known biochemical reaction mechanisms, including the law of mass-action as well as Michaelis-Menten and Hill kinetics. 

To prove the stability of a steady-state of (\ref{eq:reactions}), we must prove that the Jacobian of (\ref{eq:reactions}), evaluated at an arbitrary steady-state, always has eigenvalues with negative real part. As shown in \cite{Reznik2010} (and re-derived in the Appendix, Equations (\ref{eq:appendixS}-\ref{eq:appendixtheta})), the Jacobian $J_n$ for the metabolic cycle of size $n$ illustrated above can be calculated using SKM to be

\begin{align}
J_{n} =
\begin{vmatrix}
-\theta_{1} & 0 & ... & 0 & \theta_{n} & -\theta_{n+2} \\
\theta_{1} & -\theta_{2} & ... & 0 & 0 & \theta_{n+2} \\
0 & \theta_{2} & \ddots & \vdots & \vdots & \vdots \\
\vdots & \vdots & \ddots & -\theta_{n-1} & 0 & 0 \\
0 & 0 & ... & \theta_{n-1} & -\theta_{n} - \theta_{n+1} & 0 \\
-\theta_{1} & 0 & ... & 0 & 0 & -\theta_{n+2} - \theta_{n+3}.
\end{vmatrix}\label{eq:Jgeneral}
\end{align}

Crucially, the assumption of monotonic kinetics constrains all the elasticities, $\theta_i$, in (\ref{eq:Jgeneral}) to be greater than zero (see Appendix). In \cite{Reznik2010}, we provided evidence $J_n$ could not have eigenvalues with positive real part. Below, we proceed to analytically prove this conjecture. To simplify some calculations, we elect to work with the negative counterpart of the Jacobian, $J_n^- = -J_n$, and prove that $J_n^-$ cannot have eigenvalues with negative real part. This is equivalent to proving that $J_n$ cannot have eigenvalues with positive real part.

First, we calculate the characteristic polynomial of $J_n^-$, which we call $\chi_n(\lambda)$, explicitly (calculations shown in Appendix):
\begin{align}
P_n &= \left( (\theta_1 - \lambda)(\theta_{n+2} + \theta_{n+3} - \lambda) - \theta_{1}\theta_{n+2} \right)(\theta_{n} + \theta_{n+1}-\lambda)\prod_{i = 2\hdots n-1}{(\theta_i-\lambda)} \label{eq:Pn} \\
Q_1 &= -(\theta_{n+3} - \lambda)\prod_{i = 1\hdots n}{\theta_i} \\
\chi_{n}(\lambda) &= P_n + Q_1
\label{eq:definitions}
\end{align}

\noindent Thus, $\chi_n$ is the sum of two polynomials, an $n+1^{th}$ order polynomial $P_n$ and a first-order polynomial $Q_1$ (note that the subscript $n$ denotes the size of the cycle, not the degree of the polynomial). Next, we prove three lemmas on the relative location of the roots of $P_n$ and $Q_1$, showing that they are strongly constrained. Later on, these constraints will be crucial to proving that $\chi_n$ can only have roots with positive real part.

\begin{lemma}{All of the roots of $P_n$ are positive and real.}
\end{lemma}

\begin{proof}

By inspection, at least $n-1$ of $P_n$'s roots, contained in the product term of (\ref{eq:Pn}), must be positive and real. For the remaining two roots, we must study the quadratic polynomial 

\begin{align}
(\theta_1 - \lambda)(\theta_{n+2} + \theta_{n+3} - \lambda) - \theta_{1}\theta_{n+2}.
\label{eq:lemma1_1}
\end{align}

First, we prove that the roots of (\ref{eq:lemma1_1}) must be real. Calculating the discriminant $\Delta$ of this quadratic polynomial, we find

\begin{align}
\Delta &= (\theta_1 + \theta_{n+2} + \theta_{n+3})^2 - 4\theta_1\theta_{n+3} \nonumber \\
&= \theta_1^2 + \theta_{n+2}^2 + \theta_{n+3}^2 + 2\theta_{1}\theta_{n+2} + 2\theta_1\theta_{n+3} + 2\theta_{n+2}\theta_{n+3} - 4\theta_1\theta_{n+3} \nonumber \\
&= \theta_1^2 - 2\theta_1\theta_{n+3} + \theta_{n+3}^2 + \theta_{n+2}^2 + 2\theta_{n+2}\theta_{n+3} + 2\theta_1\theta_{n+2} \nonumber \\
&= (\theta_1 - \theta_{n+3})^2 + \theta_{n+2}^2 + 2\theta_{n+2}\theta_{n+3} + 2\theta_1\theta_{n+2} \nonumber \\
&> 0. \nonumber
\end{align}

Since $\Delta>0$, (\ref{eq:lemma1_1}) cannot have imaginary roots and all of the roots of $P_n$ are purely real.

Next, we show that the pair of roots of (\ref{eq:lemma1_1}) must be positive. If we expand (\ref{eq:lemma1_1}), we find 
\begin{align}
\lambda^2 - \lambda(\theta_1+\theta_{n+2}+\theta_{n+3}) + \theta_1\theta_{n+3} = 0.
\label{eq:lemma1_2}
\end{align}

The product of the two roots of (\ref{eq:lemma1_2}) are $\theta_1\theta_{n+3} > 0$, and the sum of the roots is $\theta_1+\theta_{n+2}+\theta_{n+3} > 0$. Therefore, both of the roots of (\ref{eq:lemma1_2}) must both be positive. 

\end{proof}

\noindent We next prove a related lemma regarding the location of the root of $Q_1$.

\begin{lemma}{The root $r_q$ of $Q_1$ must be larger than at least one root of $P_n$, and smaller than another root of $P_n$.}
\end{lemma}

\begin{proof}
Consider the quadratic factor of $P_n$, $(\theta_1 - \lambda)(\theta_{n+2} + \theta_{n+3} - \lambda) - \theta_{1}\theta_{n+2}$. By Lemma 1, this quadratic polynomial has distinct (since $\Delta > 0$) real, positive roots. Let $p_1, p_2$ denote these roots, ordered by magnitude so that $p_1<p_2$. Since the leading term of the quadratic is positive, the roots of the polynomial divide the real line into 3 regions: $\{\lambda \leq p_1\}$ and $\{\lambda \geq p_2\}$ where the polynomial is greater than or equal to zero, and $\{p_1 < \lambda < p_2\}$ where the  polynomial is strictly negative. By inspection, $r_q = \theta_{n+3}$. Directly evaluating the value of the quadratic polynomial at $\lambda = r_q = \theta_{n+3}$, we find:

\begin{align}
\theta_{n+3}^2 - \theta_{n+3}(\theta_1+\theta_{n+2}+\theta_{n+3}) + \theta_1\theta_{n+3} = -\theta_{n+3}\theta_{n+2} < 0.
\end{align}

Because the value of the quadratic factor is negative at $\lambda = r_q$, $r_q$ must lie between the roots of (\ref{eq:lemma1_2}).

\end{proof}

Finally, we prove a lemma regarding the magnitude of $P_n$ and $Q_1$ at the origin.

\begin{lemma}{$|P_n(0)| > |Q_1(0)|$}
\end{lemma}

\begin{proof}
Explicitly calculating $P_n(0)$, we find $|P_n(0)| =  |(\theta_n + \theta_{n+1}) \theta_{n+3}\prod_{i = 1\hdots n-1}\theta_i | $. This is always greater than $|Q_1(0)| = |\theta_n \theta_{n+3}\prod_{i = 1\hdots n-1}\theta_i|$.

\end{proof}

Now, using Lemmas 1-3, we proceed to prove that the roots of $\chi_n$ must lie in the positive real half of the complex plane. To do so, we will make use of a well-known theorem from complex analysis known as Rouch\'{e}'s Theorem.

\begin{theorem}[Symmetric Rouch\'{e}'s Theorem] 
Two holomorphic functions $f$ and $g$ have the same number of roots within a region bounded by some continuous closed contour $C$ (on which neither $f$ nor $g$ have any poles or zeros) if the strict inequality

\[|f(z) - g(z)| < |g(z)|
\]

holds on $C$ \cite{BrownJamesWardChurchill1996}.

\end{theorem}

In essence, Rouch\'e's theorem offers a way to determine the number of roots of a difference of two functions lying inside a closed contour in the complex plane. We will use $f = S = P_n +Q_1$ and $g = P_n$, By taking contours bounding larger and larger regions of the left half-plane (those complex numbers with negative real-part), we will show that $|Q_1| < |P_n|$ on the contour, and thus prove that $S$ has no roots inside this contour, i.e. no roots in the left half-plane. \\
\\
We let our contour $C_R$ consist of two parts (see Figure 1):
\begin{itemize}
	\item the portion of the circle $\{|\lambda| = R\}$ centered at the origin in the negative real half of the complex plane
	\item the portion of the imaginary axis connecting the two points of intersection of the above circle with the imaginary axis
\end{itemize}

First, we will prove that along an arc of sufficiently large radius, $\frac{|P_n|}{|Q_1|}>1$. Since $n>0$, given an arc of sufficiently large radius $R$, $|Q_1| < |P_n|$ simply because $P_n$ is of higher order than $Q_1$ (\textit{i.e.} the highest-order term in $P_n$ is $\lambda^{n+1}$ which dominates the highest order term in $Q_1$, $\lambda^1$, for very large $\lambda$).

Next, we will prove that along the upper half of the imaginary axis, $\frac{|P_n|}{|Q_1|}>1$.  To do so, let us consider the behavior of $\frac{|P_n|}{|Q_1|}$ by substituting $\lambda = iy, y>0 $ into (\ref{eq:definitions}) and taking the modulus. Denoting the roots of $P_n$ as $r_i$ for $i = 1, \ldots, n+1$ and the root of $Q_1$ as $r_q$, we have

\begin{align}
|P_n(iy)| &= |(r_1-iy)(r_2-iy)\hdots(r_{n+1}-iy)| \nonumber \\
&= \sqrt{(r_1^2+y^2)(r_2^2+y^2)\hdots(r_{n+1}^2+y^2)}
\end{align}

\begin{align}
|Q_1(iy)| &= |c-iby| \nonumber \\
&= \sqrt{c^2+b^2y^2} = b\sqrt{r_q^2+y^2},
\end{align}
where $c = |Q_1(0)|$ and $r_q = \frac{c}{b}$.
We must show that the following condition holds for all $y\geq 0$:
\begin{align}
x(y) = \frac{|P_n|}{|Q_1|} = \frac{\sqrt{(r_1^2+y^2)(r_2^2+y^2)\hdots(r_{n+1}^2+y^2)}}{b\sqrt{r_q^2+y^2}} > 1.
\label{eq:condition1}
\end{align}

Note that at $y = 0$, we know (\ref{eq:condition1}) is satisfied because $|P_n(0)| > |Q_1(0)|$. If we can show that $x(y)$ is a strictly increasing function, then we know that (\ref{eq:condition1}) will be satisfied for all $y$. To do so, let us work with $x^2(y)$. Note that $x(y)>0,x^2(0) > 1$, and if $\frac{d}{dy}x^2 > 0$ for all $y$, then $x^2(y) > 1$ for all $y$. This would then imply that $x(y) > 1$ for all $y \geq 0$. We have

\begin{align}
x^2(y) = \frac{|P_n|^2}{|Q_1|^2} = \frac{(r_1^2+y^2)(r_2^2+y^2)\hdots(r_{n+1}^2+y^2)}{b^2(r_q^2+y^2)}.
\end{align}

Taking a derivative of $x^2$ with respect to $y$ and using the identity $\frac{d}{dy} (f_1f_2) = \frac{df_1}{dy}f_2 +\frac{df_2}{dy}f_1$, where $f_1 = |P_n|^2, f_2 = \frac{1}{|Q_1|^2}$ (the product rule applied to $|P_n|^2$ and $1/|Q_1|^2$), we find 

\begin{align}
b^2 \frac{d}{dy} \left(x^2\right) = \frac{2y}{r_q^2+y^2}\left(\sum_{i = 1\hdots n+1}{\left(\prod_{j\neq i}{(r_j^2+y^2)} \right) }  - \frac{1}{r_q^2+y^2} \prod_{k = 1\hdots n+1} {(r_k^2 + y^2)}\right).
\label{eq:dydl}
\end{align}

Without loss of generality, suppose $r_1$ is the smallest root of $P_n$. From the first term on the right hand side of (\ref{eq:dydl}), select the term corresponding to $i = 1$. Recall that, by Lemma 2, $r_1$ must be smaller than $r_q$. Isolating just this term and the negative term in the parentheses on the right-hand-side of (\ref{eq:dydl}) and summing, we find

\begin{align}
\left( 1 - \frac{r_1^2+y^2}{r_q^2+y^2} \right) \prod_{i = 2\hdots n+1} {r_i^2 + y^2}.
\label{eq:subtract}
\end{align}

Since $r_q > r_1$, (\ref{eq:subtract}) is positive. There are no more negative terms in (\ref{eq:dydl}), proving that $x^2(y)$ is strictly increasing. This proves that $|P_n| > |Q_1|$ on the positive imaginary axis. Furthermore, since real polynomials are symmetric across the real axis, an identical argument shows that $|P_n| > |Q_1|$ on the negative imaginary axis. In particular, setting $\lambda = - i y$ for $y > 0$ yields the exact same expressions for $|P_n|$ and $|Q_1|$.

We have satisfied all of the assumptions of Rouch\'e's Theorem, and have proven $\chi_n(\lambda)$ contains no roots in the left half of the complex plane. Therefore, the nonautocatalytic metabolic cycle always has stable equilibria.

\subsection*{The Stubborn Roots Theorem}
Can we use the methods illustrated in the prior section to locate the roots of the sum of two more general polynomials? Our motivation for studying this problem derives from control theory, where it is common to ask whether the roots of a polynomial lie in one half of the complex plane \cite{Astrom2008}. Such polynomials frequently correspond to the characteristic equation of the Jacobian matrix of a dynamical system. Although we do not provide a generic method for predicting whether a matrix's characteristic equation may be written as the sum of two simpler polynomials, the appearance of such structure in our studies of a metabolic cycle suggests that related, ``well-ordered'' systems may exhibit similar properties.

The main question we ask in this section is under what conditions may the roots of a polynomial $P$ be ``stubborn'' when $P$ is summed with another polynomial $Q$: in such a case, the roots of the summed polynomial $S = P+Q$ remain qualitatively identical to those of $P$. By qualitatively identical, we mean specifically that the number of roots of $P$ in the left (right) half of the complex plane is equal to the number of roots of $S$ in the left (right) half of the complex plane. This question follows in the spirit of similar work by Anderson \cite{Anderson1993}. Our primary result is a theorem, which we call the Stubborn Roots Theorem, which gives sufficient conditions under which the location of the roots of a sum of polynomials $S = P + Q$ remains qualitatively unchanged from $P$. 

\begin{theorem}[Stubborn Real Roots Theorem] 
Let $P_n$ and $Q_m$ be polynomials of order $n$ and $m$, respectively, and let $ n>m$. Assume that all the roots of $P_n$ and $Q_m$ are purely real, and that $|P_n(0)| > |Q_m(0)|$. Denote by $p_i$ and $q_i$ the roots of $P_n$ and $Q_m$ ordered by magnitude, so that $|p_1|<|p_2|<...<|p_n|$ and $|q_1|<|q_2|<\hdots<|q_m|$. If for every $j = 1...m$, $|p_j| < |q_j|$, then the number of roots of $S = P_n + Q_m$ located in the negative (positive) real half of the complex plane is equal to the number of roots of $P_n$ in the negative (positive) real half of the complex plane.
\end{theorem} 

The proof of Theorem 2 is provided in the Appendix, and follows precisely the same line of reasoning as the proof of the cycle stability conjecture in the previous section. The theorem relies on two critical assumptions relating $P_n$ and $Q_m$. First, at the origin, $|P_n(0)| > |Q_m(0)|$. Second, it must be possible to assign to each root of $q_i$ of $Q_m$ a unique root $p_i$ of $P_n$ such that $q_i > p_i$. 

The power of Theorem 2 is that it enables one to qualitatively locate the roots of a polynomial $S$ simply by inspecting the roots of its summands $P_n$ and $Q_m$. If the roots of $P_n$ and $Q_m$ are easily calculated (as in the case of the cycle stability conjecture in the prior section), then the roots of $S$ can be immediately located without resorting to difficult calculations. In many ways, Theorem 2 is reminiscent of the work reported in \cite{Fisk2006}. There, Fisk describes the behavior of the roots of sums of polynomials which ``interlace.'' For two polynomials $P$ and $Q$ to interlace, the roots of $P$ and $Q$ alternate when ordered from most negative to most positive, so that $p_1<q_1<p_2<q_2 \hdots $. Notably, our result here is more general, and includes interlacing as a special case.

\subsection*{Stubborn Complex Roots}
What happens when matters become complex? In this section, we generalize the Stubborn Roots Theorem to cases when the roots of $P$ are not necessarily all real.  This is often the case in dynamical systems, where complex roots indicate oscillatory phenomena such as spiraling or limit cycles \cite{Strogatz1994}. Proceeding along the same lines as before, we find that the roots of $P$ remain stubborn to the addition of $Q$ as long as they remain predominantly real. That is, if the real component of the complex roots of $P$ is larger than their imaginary component, then a more general version of the Stubborn Roots Theorem holds.

\begin{theorem}[Stubborn Complex Roots Theorem] 
Let $P_n$ and $Q_m$ be polynomials of order $n$ and $m$, respectively, $n > m$. Let the $m$ roots of $Q_m$ be positive and purely real. Further, let $P_n$ have at least $m$ real roots, and let the remainder of the roots be either real or complex. Assume that $|P_n(0)| > |Q_m(0)|$. Furthermore, assume that for each complex root of $P$, $p_k$, the magnitude of the real component $|Re(p_k)|$ is larger than the magnitude of its imaginary component, $|Im(p_k)|$. Denote by $p_i$ and $q_i$ the real roots of $P_n$ and $Q_m$ ordered by magnitude, so that $|p_1|<|p_2|<...<|p_m|$ and $|q_1|<|q_2|<\hdots<|q_m|$. If for every $j = 1...m$, $|p_j| < |q_j|$, then the number of roots of $S = P_n + Q_m$ located in the negative (positive) real half of the complex plane is equal to the number of roots of $P_n$ in the negative (positive) real half of the complex plane (Figure 2).
\end{theorem}

\begin{proof}
We prove the theorem for the case when $P_n$ has $n-2$ real roots, two complex conjugate roots $p_c^+ = a+bi$ and $p_c^- = a-bi$, and $Q$ has one positive real root. The result can be straightforwardly (via wrenching and tedious algebraic calculations) extended to the generic case in Theorem 3 using an identical argument. As before, we apply Rouch\'{e}'s Theorem using a half circle in the negative real half of the complex plane using $f = P = P_n +Q_1$ and $g = P_n$. First, we consider the behavior of the two polynomials on the large arc in the negative real half of the complex plane. As before, $|P_n|$ dominates $|Q_1|$ as the radius of the arc grows larger. 

Turning our attention to the behavior of the polynomials on the positive imaginary axis, we substitute $z = iy$ to find

\begin{eqnarray}
b^2\frac{|P_n|^2}{|Q_1|^2} = \frac{\prod_{i=1}^{n-2}(y^2+p_i^2)((y-b)^2+a^2)((y+b)^2+a^2)}{y^2+q_1^2} ,
\end{eqnarray}
where for simplicity of notation, we have assumed $P_n$ has leading coefficient $1$.
Differentiating and simplifying algebraically, we obtain the expression
\begin{align}
(y^2+q_1^2)b^2 \frac{d}{dy} \left(\frac{|P_n|^2}{|Q_1|^2} \right) &= C \left( \sum_{i=1}^{n-2}2y \left( \prod_{j=1,j\neq i}^{n-2} p_j^2+y^2 \right) -\frac{2y}{q_1^2+y^2}\left( \prod_{j=1}^{n-2} p_j^2+y^2 \right) \right) \nonumber \\
& + \left(2(y-b)((y+b)^2+a^2) + 2(y+b)((y-b)^2+a^2) \right)\prod_{i=1}^{n-2}(p_i^2+y^2) ,
\label{eq:complex1}
\end{align}

where $C = ((y-b)^2+a^2)((y+b)^2+a^2) $. Using an argument identical to the one used to prove the cycle stability conjecture, it is clear that the top term in (\ref{eq:complex1}) is positive. We are then left to ensure that the bottom term is positive. Expanding this term, we find
\begin{align}
& \left(2(y-b)((y+b)^2+a^2) + 2(y+b)((y-b)^2+a^2) \right)\prod_{i=1}^{n-2}(p_i^2+y^2)  \nonumber \\
&= \left(4y(y^2+a^2+b^2)+2b(-4by) \right) \prod_{i=1}^{n-2}(p_i^2+y^2)  \nonumber \\
&= \left( 4y(y^2+a^2-b^2) \right) \prod_{i=1}^{n-2}(p_i^2+y^2).
\label{eq:complex2}
\end{align}

Thus, (\ref{eq:complex2}) is certain to be positive as long as $|a|>|b|$. In this case, we can once more apply Rouch\'{e}'s theorem to show that $S = P_n+Q_1$ has the same number of roots in each half of the complex plane as $P_n$.

\end{proof}

\section*{Discussion}
A major challenge in systems biology is efficiently studying biological networks through detailed atlases of their topological structures. These atlases, assembled from the cumulative results of many high-throughput, large-scale experiments, capture many of the physical links which underlie such networks. However, they fail to describe most of the detailed dynamics taking place on the network itself. Here, we have studied the stability properties of a generic type of metabolic network, and analytically proven that under quite mild assumptions on the reaction kinetics, any steady-state of the network must be stable. To prove this, we re-formulated the question of stability as a problem of locating the roots of a sum of two polynomials whose roots were easily calculated. This reformulation exposed a more fundamental problem of locating the roots of the sum of two polynomials, and we proved two new results (the Stubborn Roots Theorems) offering sufficient conditions under which the roots of the sum are not qualitatively different from the roots of one of the summands.

The study of metabolic networks and their stability has played an important role in research into the origin of life. Many of the earliest papers studying simple models of primordial metabolic networks focused on elucidating their stability properties, hypothesizing that molecular self-organization may arisen through self-sustaining metabolic cycles \cite{Eigen1978}. Because early metabolism almost certainly lacked the complex regulatory mechanisms and circuits which appear in cells today, these investigations suggested that stability of these cycles to fluctuations in environmental conditions was necessary for their survival. Pursuing this line of thought, Piedrafita \textit{et al} recently proposed a simple chemical reaction network composed of interlocking cycles which could establish and maintain a stable steady-state, even in the face of a sudden loss of some constituent metabolite of the cycle itself \cite{Piedrafita2010}. In response to such a catastrophe, the remaining metabolites re-produced the missing metabolite. This notion of \textit{closure}, in which all of the metabolites necessary for sustaining the system can be produced by the metabolic network itself, also plays an important role in Chemical Organization Theory, a method distinct from SKM for studying chemical reaction networks based on structure and mentioned earlier \cite{Dittrich2007,Kreyssig2012}.

Naturally, one may ask: how important is the stability of equilibria to the robust function of \textit{present-day} biological networks? A rich and diverse literature, dating back nearly half a century and still expanding today, describes the importance of stabilizing structures in ecological networks \cite{May1973, Gross2009}. However, the extent to which stability plays a role in the fitness of metabolic systems is unclear, and studies investigating this question have failed to produce a definitive conclusion. In particular, if stability endowed metabolic networks with some evolutionary advantage, one should expect an enrichment for stabilizing features and structures in contemporary metabolic networks. While some studies have demonstrated the importance of some key stabilizing edges in metabolic networks, such as the allosteric feedback of ATP onto phosphofructokinase in glycolysis \cite{Steuer2006,Gehrmann2011}, others have failed to identify an enrichment of stabilizing structures in metabolism as a whole \cite{vanNes2009}. In fact, synthetic biologists routinely exploit \textit{instability} in order to generate circuits exhibiting sustained oscillations, both in metabolic \cite{Fung2005} and transcriptional \cite{Elowitz2000} systems. Importantly, we note that the work presented here only studied dynamics near equilibrium points, and ignored nonlocal dynamics (such as the appearance of periodic orbits arising from global bifurcations). Efforts extending generalized modeling and SKM to understanding nonlocal dynamics, are now appearing in the literature \cite{Kuehn2013}. 

We expect that the results presented here may find useful application in several challenges facing contemporary biology. First, SKM and generalized modeling could be used as coarse-grained techniques for vetting synthetic circuit designs for their potential to exhibit desirable behaviors (such as robust stability). In a prior study \cite{Reznik2013}, we did precisely this, using generalized modeling to identify which topological circuit designs were entirely incapable of oscillations, irrespective of the choice of kinetic parameters or rate laws. Second, a great deal of interest now exists in using high-throughput metabolomics and fluxomics data to identify the role that small-molecule (\textit{e.g.} allosteric) regulation of metabolic enzymes plays in shaping metabolic dynamics \cite{Link2013}. Given the difficulty in accurately measuring kinetic rate constants \textit{in vivo}, we envision that SKM (and our results here highlighting the inherent stability of certain topological motifs) might serve as a useful bridge between detailed mechanistic models and experimental data, highlighting those regulatory interactions which are crucial to the robust function of the network as a whole.

Finally, our results (in particular, the Stubborn Roots Theorems) illustrate the potential for biological questions to reveal interesting and unsolved problems in other fields. What appeared to us initially as a simple problem of locating the roots of the sum of two stable polynomials, quickly blossomed into the exploration of widely diverse fields of active research, from control theory to matrix analysis. There is now a growing number of examples of similar feedback from biology to other fields, from classic results in evolutionary optimization \cite{Fogel2006} to to the design of novel algorithms \cite{Afek2011}. Interestingly, many of these cross-fertilizations of ideas have taken place because of an abundance of biological data, and a need for analytical tools to understand it. Here, it has been quite the contrary: our study of a topological model of a metabolic cycle was motivated by a dearth of data on the kinetics of metabolic reactions. Nevertheless, in both cases, the ultimate outcome is deeper understanding, relevant to both biology and the fields from which it draws new tools and ideas.

\section*{Appendix}
\subsection*{Structural Kinetic Modeling}
To study the stability of a steady-state of a metabolic network, we employ a technique known as structural kinetic modeling (SKM) \cite{Steuer2006}.  SKM is a non-dimensionalization procedure which replaces conventional kinetic parameters (such as $V_{max}$ and $K_M$) with normalized parameters known as elasticities. In the past, SKM (and its generalization, known as Generalized Modeling (GM)) has been paired with complementary methods studying other dynamic features of a system, such as the effects of noise \cite{Gehrmann2011}.

As illustrated below, elasticities have several properties which make them powerful tools for studying metabolic dynamics. First, in contrast with kinetic parameters (whose values may be uncertain over many orders of magnitude), elasticities are constrained to lie in well-defined ranges (for example, between zero and one), and sampling elasticities across this range effectively captures all possible values of kinetic parameters. Second, the value of an elasticity \textit{does not depend} on the particular choice of kinetic rate law. Instead, an elasticity is simply a normalized measure of the sensitivity of a rate law to infinitesimal changes in a metabolite's concentration. 

To study the stability of a steady state, SKM calculates the Jacobian matrix $J$ of the dynamical system corresponding to a metabolic network. If we let $\mathbf{C}$ be the $m$-dimensional vector of metabolite concentrations, $\mathbf{N}$ be the $m \times r$ stoichiometric network, and $\mathbf{v}$ be the $r$-dimensional vector of metabolic fluxes, then the dynamics of a metabolic network are governed by the system of differential equations 

\begin{align}
\frac{d\mathbf{C}}{dt} = \mathbf{Nv}(\mathbf{C,k}),
\end{align}

\noindent where $\mathbf{k}$ is a vector of parameters and $\mathbf{v}(\mathbf{C,k})$ indicates that the vector of fluxes is dependent on both metabolite concentrations and kinetic parameters. Assuming that a non-zero steady-state $\mathbf{C^0}$ exists, we can make a change of variables and write

\begin{align}
x_i &= \frac{C_i}{C_i^0} \hspace{1cm} \Lambda_{ij} = N_{ij}\frac{v_j(\mathbf{C^0},\mathbf{k})}{C_i^0} \hspace{1cm} \mu_j(\mathbf{x}) = \frac{v_j(\mathbf{C},\mathbf{k})}{v_j(\mathbf{C^0},\mathbf{k})},
\label{eq:SKM1}
\end{align}

\noindent where $i = 1 \hdots m$ and $j=1 \hdots r$. 

Then, we can write the Jacobian as 
\begin{align}
\mathbf{J} = \mathbf{\Lambda} \frac{\partial \mu_j(\mathbf{C})}{\partial C_i}  = \mathbf{\Lambda \Theta}.
\label{eq:SKM2}
\end{align}

\noindent The stability of the steady-state $\mathbf{C^0}$ is then dependent on the eigenvalues of $\mathbf{J}$. If the real component of all eigenvalues of $\mathbf{J}$ are negative, then the steady-state is stable. Thus, the problem of stability reduces to finding the eigenvalues of $\mathbf{J}$.

The element which encodes the effective kinetic dependence of reaction rates on metabolites of SKM is the $r \times m $ \textit{elasticity matrix} $\mathbf{\Theta}$. The $(i,j)^{th}$ element of $\mathbf{\Theta}$ describes the sensitivity of the normalized rate of reaction $i$ to the normalized concentration of metabolite $j$. This corresponds precisely to the \textit{effective kinetic order} of the $i^{th}$ reaction with respect to the $j^{th}$ substrate: if the rate of reaction is linear with the amount of substrate, then $\theta = 1$, while if it is zeroth order, $\theta = 0$ \cite{Fell1992}. Importantly, we assume that all elasticities in the metabolic cycle are greater than zero; that is, that an increase in the substrate of any reaction will increase the rate of that reaction. The analytical power of SKM comes precisely from the constrained and well-defined ranges of each element of $\mathbf{\Theta}$. 

To illustrate the utility of elasticities, we derive below the elasticity of a metabolite involved in a Michaelis-Menten reaction. Consider a biochemical reaction governed by the rate law

\begin{align}
v = \frac{-V_{max}S}{K_M+S}.
\label{eq:MM1}
\end{align}

Assuming that this reaction is embedded within a reaction network where metabolite $S$ is at equilibrium concentration $S_0$, we can calculate the normalized reaction rate $\mu$ by normalizing (\ref{eq:MM1}) by its steady-state reaction rate:

\begin{align}
\mu &= \frac{ \frac{-V_{max}S}{K_M+S} } {\frac{-V_{max}S_0}{K_M+S_0} } \nonumber \\
&= \frac{S(K_M+S_0)}{S_0(K_M+S)} \nonumber \\
& = \frac{x(K_M+S_0)}{K_M+xS_0},
\label{eq:MM2}
\end{align}

where $x = S/S_0$ is the normalized concentration of $S$. Then, the elasticity is
\begin{align}
\theta = \frac{ \partial \mu}{\partial x }\bigg|_{x=1} = \frac{1}{1+x}\bigg|_{x=1} = \frac{K_M}{K_M+S_0}.
\label{eq:MM3}
\end{align}

Notice that since $S_0 > 0$, $\theta$ is constrained to the range $(0,1)$. The outcome of applying SKM to an entire metabolic network is a Jacobian, whose elements are formulated in terms of elasticities with well-defined ranges. Prior studies have used computational surveys \cite{Steuer2006, Grimbs2007} as well as analytical work \cite{Reznik2010} to study the role that particular key elasticities play in determining the stability of the network. 

Finally, it may be useful to give a bit more intuition regarding the generality of an elasticity. To do so, we consider a reaction, governed by Michaelis-Menten kinetics, which exhibits an elasticity (explicitly calculated in Equation (\ref{eq:MM3}) above) of 0.5. First, note that this elasticity may correspond to any combination of $S_0$ and $K_M$ which satisfy $\frac{K_M}{K_M+S_0} = 0.5$, for example $K_M = S_0 = 1$ or $K_M = S_0 = 2$. Thus, a single value for an elasticity in fact corresponds to a large locus of steady-state concentrations $S_0$. Now, notice that if we consider all elasticities in the range $(0,1)$, we in fact capture all possible combinations of $S_0>0$ and $K_M>0$! Thus, if we prove a theorem using this elasticity, and this theorem holds for all values of the elasticity in the range $(0,1)$, then it similarly holds for all possible choices of $S_0$ and $K_M$. Obviously, the result also holds for any other kinetic rate laws for which this elasticity is valid. Thus, we effectively capture the entire space of possible parameters and steady-state concentrations. This is precisely the approach taken in proving the cycle stability conjecture.

\subsection*{Characteristic Polynomial for the General Metabolic Cycle}
\label{sec:appendix_jacobian}

Here, we formulate the structural kinetic model for the metabolic cycle depicted in (\ref{eq:reactions}) and reproduced below. 

\begin{align}
\emptyset &\longrightarrow M_{1} \nonumber \\
M_{1} + O_{1} &\longrightarrow M_{2} + O_{2} \nonumber \\
M_{i} &\longrightarrow M_{i+1}, \hspace{2mm}  i = 2\hdots n-1 \nonumber \\
M_{n} &\longrightarrow M_{1} \nonumber \\
M_{n} &\longrightarrow \emptyset \nonumber \\
O_2 &\xrightarrow{energy} O_1 
\label{eq:reactions_appendix}
\end{align}

The analysis presented below is identical to that presented in \cite{Reznik2010}. The stoichiometric matrix $\mathbf{S}$ of the metabolic cycle is

\begin{equation}
\mathbf{S} = \left[
\begin{array}{ccccccccc}
1 & -1 & 0  & \cdots & 0 & 0 & 1 & 0\\
\vspace{1pt} \\
0 & 1 & -1 & \cdots & 0 & 0 & 0 & 0\\
\vspace{1pt} \\
0 & 0 & 1 & \cdots & 0 & 0 & 0 & 0\\
\vspace{1pt} \\
\vdots & \vdots & \vdots & \vdots & \vdots & \vdots & \vdots & \vdots \\
\vspace{1pt}\\
0 & 0 & 0 & \cdots & -1 & 0 & 0 & 0 \\
\vspace{1pt}
0 & 0 & 0 & \cdots & 1 & -1 & -1 & 0 \\
\vspace{1pt}
0 & -1 & 0 & \cdots & 0 & 0 & 0 & 1 \\
0 & 1 & 0 & \cdots & 0 & 0 & 0 & -1
\end{array}
\right],
\label{eq:appendixS}
\end{equation}

\noindent where the rows correspond to each metabolite in the system. The generalized forms of $\mathbf{\Lambda}$, the normalized stoichiometric matrix, and $\mathbf{\Theta}$, the elasticity matrix, are shown below for the system depicted in (\ref{eq:reactions}). Note that we make no assumptions on the steady state concentrations of metabolites, denoted by the vector $(M_1, M_2,\dots,M_N,O_1,O_2)$. Furthermore, the steady-state flux through the cycle is equal to a generic magnitude $v, v>0$, except for the input and outflow reactions (with flux $\alpha v$) and the reaction from $M_n$ to $M_1$ (with flux $(1-\alpha)v$). 

\begin{equation}
\mathbf{\Lambda} = \left[
\begin{array}{ccccccccc}
\frac{\alpha v}{M_1} & \frac{-v}{M_1} & 0  & \cdots & 0 & 0 & \frac{(1-\alpha)}{M_1} & 0\\
\vspace{1pt} \\
0 & \frac{v}{M_2} & \frac{-v}{M_2} & \cdots & 0 & 0 & 0 & 0\\
\vspace{1pt} \\
0 & 0 & \frac{v}{M_3} & \cdots & 0 & 0 & 0 & 0\\
\vspace{1pt} \\
\vdots & \vdots & \vdots & \vdots & \vdots & \vdots & \vdots & \vdots \\
\vspace{1pt}\\
0 & 0 & 0 & \cdots & \frac{-v}{M_{N-1}} & 0 & 0 & 0 \\
\vspace{1pt}
0 & 0 & 0 & \cdots & \frac{v}{M_{N}} & \frac{-\alpha v}{M_N} & -\frac{(1-\alpha)v}{M_N} & 0 \\
\vspace{1pt}
0 & \frac{-v}{O_1} & 0 & \cdots & 0 & 0 & 0 & \frac{v}{O_1}
\end{array}
\right]
\label{eq:appendixlambda}
\end{equation}

and
\begin{equation}
\mathbf{\tilde{\Theta}} = \left[
\begin{array}{cccccc}
0 & 0 & 0 & \cdots & 0 & 0 \\
\tilde{\theta}_1 & 0 & 0 & \cdots & 0 & \tilde{\theta}_{N+2} \\
0 & \tilde{\theta}_2 & 0 & \cdots & 0 & 0 \\
0 & 0 & \tilde{\theta}_3 & \cdots & 0 & 0 \\
0 & 0 & 0 & \cdots & \tilde{\theta}_N & 0\\
0 & 0 & 0 & \cdots & \tilde{\theta}_{N+1} & 0\\
0 & 0 & 0 & \cdots & 0 & -\frac{\tilde{\theta}_{N+3}O_1}{O_2}
\end{array}
\right].
\label{eq:appendixtheta}
\end{equation}
\noindent With $N+1$ metabolites and $N+3$ reactions, $\mathbf{\Lambda}$ is $N+1 \times (N+3)$ and $\mathbf{\Theta}$ is $(N+3) \times N+1$. Note that the last row (corresponding to cofactor $O_2$) of $\mathbf{\Lambda}$ is omitted because the cofactors come as a conserved pair, and the bottom right element of $\mathbf{\tilde{\Theta}}$ (corresponding to the dependence of the last reaction on $O_2$) is replaced with a negative element in order to account for this conservation (for more information on modeling of conserved moeties, see the SI Text of \cite{Steuer2006}). Then, the Jacobian $J_n = \mathbf{\Lambda \tilde{\Theta}}$. We elect to study the eigenvalues of the negative of $J_n$, which we call $J_n^-$. Note that the eigenvalues of $J_n^-$ are precisely the negative of the eigenvalues of $J_n$. Thus, proving that all the eigenvalues of $J_n^-$ have positive real part is equivalent to proving that all of the eigenvalues of $J_n$ have negative real part. The characteristic polynomial of $J_n^-$, $\chi_{n}(\lambda)$ is

\[
\chi_{n}(\lambda) =
\begin{vmatrix}
\frac{\theta_{1}}{M_1} - \lambda & 0 & ... & 0 & \frac{-\theta_{n}}{M_1} & \frac{\theta_{n+2}}{M_1} \\
\frac{-\theta_{1}}{M_2} & \frac{\theta_{2}}{M_2} - \lambda & ... & 0 & 0 & \frac{-\theta_{n+2}}{M_2} \\
0 & \frac{-\theta_{2}}{M_3} & ... & ... & ... & ... \\
... & ... & ... & \frac{\theta_{n-1}}{M_{n-1}} - \lambda & 0 & 0 \\
0 & 0 & ... & \frac{-\theta_{n-1}}{M_n} & \frac{\theta_{n} + \theta_{n+1}}{M_n} - \lambda & 0 \\
\frac{\theta_{1}}{O_1} & 0 & ... & 0 & 0 & \frac{\theta_{n+2} + \theta_{n+3}}{O_1} - \lambda
\end{vmatrix}
\].

Above, we have made a change of variables so that $\theta_i = \tilde{\theta}_i, i = 2...n+2$, $\theta_1 = \alpha\tilde{\theta}_1$, and $\theta_{n+3} = \frac{\tilde{\theta}_{N+3}O_1}{O_2}$. Note that this change of variables does not affect our assumption that all $\theta$ are greater than zero. 

Now, we proceed to explicitly calculate $\chi_n(\lambda)$ by calculating the determinant. Expanding along the $n-1^{th}$ column, we have:

\[
\chi_{n}(\lambda) = (\frac{\theta_{n-1}}{M_{n-1}} - \lambda)
\begin{vmatrix}
\frac{\theta_1}{M_1} - \lambda & 0 & ... & 0 & \frac{-\theta_n}{M_1} & \frac{\theta_{n+2}}{M_1} \\
\frac{-\theta_1}{M_2} & \frac{\theta_2}{M_2} - \lambda & ...& 0 & 0 & \frac{-\theta_{n+2}}{M_2} \\
0 & \frac{-\theta_2}{M_3} & ... & ... & ... & ... \\
... & ... & ... & \frac{\theta_{n-2}}{M_{n-2}} - \lambda & 0 & 0 \\
0 & 0 & ... & 0 & \frac{\theta_{n} + \theta_{n+1}}{M_n} - \lambda & 0 \\
\frac{\theta_1}{O_1} & 0 & ... & 0& 0 & \frac{\theta_{n+2} + \theta_{n+3}}{O_1} - \lambda
\end{vmatrix}
\]
\[
+ \frac{\theta_{n-1}}{M_n}
\begin{vmatrix}
\frac{\theta_1}{M_1} - \lambda & 0 & ... & 0 & \frac{-\theta_n}{M_1} & \frac{\theta_{n+2}}{M_1} \\
\frac{-\theta_1}{M_2} & \frac{\theta_2}{M_2} - \lambda & ...& 0 & 0 & \frac{-\theta_{n+2}}{M_2} \\
0 & \frac{-\theta_2}{M_3} & ... & ... & ... & ... \\
... & ... & ... & \frac{\theta_{n-2}}{M_{n-2}} - \lambda & 0 & 0 \\
0 & 0 & ... & \frac{-\theta_{n-2}}{M_{n-1}} & 0 & 0 \\
\frac{\theta_1}{O_1} & 0 & ... & 0& 0 & \frac{\theta_{n+2} + \theta_{n+3}}{O_1} - \lambda
\end{vmatrix}
\].

The important thing to note is that now, the $n-1^{th}$ column in the first matrix and the $n^{th}$ row in the second are zero except for a single entry. We can continue expanding along such columns until we arrive at

\[
\chi_{n}(\lambda) = (\frac{\theta_{n-1}}{M_{n-1}} - \lambda)...(\frac{\theta_2}{M_2} - \lambda)
\begin{vmatrix}
\frac{\theta_1}{M_1} - \lambda & \frac{-\theta_n}{M_1} & \frac{\theta_{n+2}}{M_1} \\
0 & \frac{\theta_n + \theta_{n+1}}{M_n} - \lambda & 0 \\
\frac{\theta_1}{O_1} & 0 & \frac{\theta_{n+2} + \theta_{n+3}}{O_1} - \lambda
\end{vmatrix}
\]
\[
+ \frac{\theta_{n-1}}{M_n}...\frac{\theta_2}{M_3}
\begin{vmatrix}
\frac{\theta_1}{M_1} - \lambda & \frac{-\theta_n}{M_1} & \frac{\theta_{n+2}}{M_1} \\
\frac{-\theta_1}{M_2} & 0 & \frac{-\theta_{n+2}}{M_2} \\
\frac{\theta_1}{O_1} & 0 & \frac{\theta_{n+2} + \theta_{n+3}}{O_1} - \lambda
\end{vmatrix}
\].

Finally, we have, for $n\geq 3$:

\begin{align}
\chi_{n}(\lambda) = &\left( (\frac{\theta_1}{M_1} - \lambda)(\frac{\theta_{n+2} + \theta_{n+3}}{O_1} - \lambda) - \frac{\theta_{1}\theta_{n+2}}{M_1O_1} \right)(\frac{\theta_{n} + \theta_{n+1}}{M_n}-\lambda)\prod_{i = 2\hdots n-1}{(\frac{\theta_i}{M_i}-\lambda)} \\
&- (\frac{\theta_{n+3}}{O_1} - \lambda)\prod_{i = 1\hdots n}\frac{\theta_i}{M_i} \nonumber.
\end{align}

If we make a final change of variables so that $\bar{\theta}_i = \frac{\theta_i}{M_i}, i = 1 \hdots n, \bar{\theta}_{n+1} = \frac{\theta_{n+1}}{M_n}, \bar{\theta}_{n+2} = \frac{\theta_{n+2}}{O_1}, \bar{\theta}_{n+3} = \frac{\theta_{n+3}}{O_1}$, then we arrively precisely at (\ref{eq:definitions}).

\subsection*{Characteristic Polynomial for a Simple Metabolic Cycle}
\label{sec:appendix_simplecycle}
To give the reader more intuition regarding the mechanics of SKM calculations, we briefly describe how SKM may be used to model the dynamics of a simple, nonautocatalytic metabolic cycle with two metabolites and one cofactor pair. The reactions of this cycle are

\begin{align}
\emptyset &\longrightarrow A \nonumber \\
A + O_{1} &\longrightarrow B + O_{2} \nonumber \\
B &\longrightarrow A \nonumber \\
B &\longrightarrow \emptyset \nonumber \\
O_2 &\xrightarrow{energy} O_1 
\label{eq:reactions_simplecyc}
\end{align}

Then, the stoichiometric matrix $\mathbf{S}$, the normalized stoichiometric matrix $\mathbf{\Lambda}$, and the elasticity matrix $\mathbf{\Theta}$, may be written

\begin{equation}
\mathbf{S} = \left[
\begin{array}{ccccc}
1 & -1 & 1 & 0 & 0\\
\vspace{1pt} \\
0 & 1 & -1 & -1 & 0\\
\vspace{1pt} \\
0 & -1 & 0 & 0 & 1\\
\vspace{1pt} \\
0 & 1 & 0 & 0 & -1
\end{array}
\right],
\label{eq:simplecycleS}
\end{equation}

\begin{equation}
\mathbf{\Lambda} = v \left[
\begin{array}{ccccc}
\frac{\alpha}{A^0} & \frac{-1}{A^0} & \frac{1-\alpha}{A^0} & 0 & 0\\
\vspace{1pt} \\
0 & \frac{1}{B^0} & \frac{-(1-\alpha)}{B^0} & \frac{-\alpha}{B^0} & 0\\
\vspace{1pt} \\
0 & \frac{-1}{O_1^0} & 0 & 0 & \frac{1}{O_1^0}
\end{array}
\right],
\label{eq:simplecycleL}
\end{equation}

\begin{equation}
\mathbf{\Theta} = \left[
\begin{array}{ccc}
0 & 0 & 0 \\
\vspace{1pt} \\
\theta_1 & 0 & \theta_4 \\
\vspace{1pt} \\
0 & \theta_2 & 0 \\
\vspace{1pt} \\
0 & \theta_3 & 0 \\
\vspace{1pt} \\
0 & 0 & \frac{-O_1^0}{O_2^0}\theta_5
\end{array}
\right],
\label{eq:simplecycleT}
\end{equation}
where $v$ is an arbitrary unit of flux, $\alpha \in (0,1)$, and $A^0,B^0,$ and $O_1^0$ are the steady-state concentrations of $A,B,$ and $O_1$, respectively.

Note that $\mathbf{\Lambda}$ has precisely the same structure as $\mathbf{S}$, but with one fewer row (due to the mass conservation associated with the cofactor pair, discussed in the prior section). Finally, the Jacobian for this system may be straightforwardly written
\begin{equation}
\mathbf{J} = \mathbf{\Lambda\Theta} =  v \left[
\begin{array}{ccc}
\frac{-\theta_1}{A^0} & \frac{(1-\alpha)\theta_2}{A^0} & \frac{-\theta_4}{A^0}\\
\vspace{1pt} \\
\frac{\theta_1}{B^0} & \frac{-(1-\alpha)\theta_2 - \alpha \theta_3}{B^0} & \frac{\theta_4}{B^0}\\
\vspace{1pt} \\
\frac{\theta_1}{O_1^0}& 0 & \frac{-\theta_4}{O_1^0} - \frac{-\theta_5}{O_2^0}
\end{array}
\right],
\label{eq:simplecycleJ}
\end{equation}

\subsection*{Proof of the Real Stubborn Roots Theorem}
Consider the sum of two polynomials

\begin{align}
S &= P_n + Q_m \\
P_n &= (p_1 - \lambda)(p_2-\lambda)\hdots(p_n-\lambda), n \geq 2 \\
Q_m &= (q_1 - \lambda)(q_2-\lambda)\hdots(q_m-\lambda), m < n.
\label{eq:master}
\end{align}

Above, all roots $p_i$ and $q_i$ are real and positive. Further, we assume that $|P_n(0)| > |Q_1(0)|$. We claim that $P$ cannot have any roots lying in the left half of the complex plane. To prove this, we make use of the symmetric form of Rouche's Theorem (Theorem 1 in the text)  We will use $f = P = P_n +Q_m$ and $g = P_n$. Our goal will be to show that $|Q_m| < |P_n|$ on the contour, which by Rouch\'{e}'s Theorem gives us that $P_n$ and $S$ have the same number of zeros inside the contour (which is none, since all of the roots of $P_n$ are positive real numbers). By taking contours bounding larger and larger regions contained in and tending towards the entire left-half plane of $\mathbb{C}$, we will deduce that $S$ has no roots with negative real component.  \\
\\
We let our contour $C = C_R$ consist of two parts:
\begin{itemize}
	\item a large half circle in the left-half of the complex plane
	\item the portion of the imaginary axis connecting the two intersections of the above circle with the positive imaginary axis.
\end{itemize}

First, we verify that $|Q_m| < |P_n|$ on the half circle $\{|\lambda|=R\}$ in the left-half plane. Since $P_n$ is of higher order than $Q_m$, as $R \rightarrow \infty$, $|P_n|$ dominates $|Q_m|$.  

Next, we verify that $|Q_m| < |P_n|$ on the imaginary axis. Again substituting $\lambda = i y$ for $y > 0$ into (\ref{eq:master}), we have
\begin{align}
|P_n| &= \sqrt{(y^2+p_1^2) \ldots (y^2+p_n^2)} \\
|Q_m| &= b \sqrt{(y^2+q_1^2) \ldots (y^2+q_m^2)}.
\end{align}

Then, 

\begin{eqnarray}
b^2 \frac{d}{dy} \left(\frac{|P_n|^2}{|Q_m|^2}\right) = \frac{\prod_{i=1}^{n}(y^2+p_i^2)}{\prod_{j=1}^{m}(y^2+q_i^2)}, 
\end{eqnarray}

and it follows that

\begin{align}
b^2 \frac{d}{dy} \left(\frac{|P_n|^2}{|Q_m|^2}\right)  &= \frac{ \left(\prod_{j=1}^m (y^2+q_j^2) \right) \left(\sum_{i=1}^n \left(2y \prod_{k=1,k\neq i}^n (y^2+p_k^2) \right) \right) } { (\prod_{j=1}^m (y^2+q_j^2))^2} \nonumber \\
& - \frac{ \left(\prod_{i=1}^n (y^2+p_i^2) \right) \left(\sum_{j=1}^m \left(2y \prod_{k=1,k\neq j}^m (y^2+q_k^2) \right) \right) } { (\prod_{j=1}^m (y^2+q_j^2))^2}.
\end{align}

Rewriting this, we find
\begin{align}
b^2 \frac{d}{dy} \left(\frac{|P_n|^2}{|Q_m|^2}\right)  &= \frac{1}{\prod_{j=1}^m (y^2+q_j^2)} \left( \sum_{i=1}^n \left(2y \prod_{k=1,k\neq i}^n (y^2+p_k^2) \right) - \sum_{j=1}^m \left( \left( \frac{2y}{y^2+q_j^2} \right) \prod_{i=1}^n (y^2+p_i^2) \right) \right).
\end{align}

Ignoring the common denominator in both terms and considering only the first $m$ terms of the first series, we obtain
\begin{align}
\sum_{i=1}^m \left(2y \prod_{k=1,k\neq i}^n (y^2+p_k^2) \right) - \sum_{j=1}^m \left( \left( 2y\frac{y^2+p_j^2}{y^2+q_j^2} \right) \prod_{i=1,i \neq j}^n (y^2+p_i^2) \right).
\label{eq:general_sum}
\end{align}

Now, comparing the terms of each series in order, we observe precisely the same pattern as in the main text; since $p_j < q_j$ for all $j=1\hdots m$, (\ref{eq:general_sum}) is always greater than zero. Then, $|P_n| > |Q_m|$ on the positive imaginary axis, and by symmetry also on the negative imaginary axis. Applying Rouch\'{e}'s Theorem once again, we find that the number of roots of $S = P_n + Q_m$ in the positive (negative) real half of the complex plane is identical to the number of roots of $P_n$ in the positive (negative) real half of the complex plane.

\section*{Acknowledgments}
We are grateful to Daniel Segr\`{e} and Gene Wayne for insightful feedback. ER was supported in part by the National Science Foundation (NSF DMS-0602204 EMSW21-RTG, BioDynamics at Boston University). OC was supported in part by the National Science Foundation (NSF DMS-0908093).

\bibliography{cyclebib}

\begin{thebibliography}{10}

\bibitem{Afek2011}
Yehuda Afek, Noga Alon, Omer Barad, Eran Hornstein, Naama Barkai, and Ziv
  Bar-Joseph.
\newblock {A biological solution to a fundamental distributed computing
  problem.}
\newblock {\em Science (New York, N.Y.)}, 331(6014):183--5, January 2011.

\bibitem{Anderson1993}
Bruce Anderson.
\newblock {Polynomial Root Dragging}.
\newblock {\em The American Mathematical Monthly}, 100(9):864--866, 1993.

\bibitem{Angeli2004}
David Angeli, James~E Ferrell, and Eduardo~D Sontag.
\newblock {Detection of multistability, bifurcations, and hysteresis in a large
  class of biological positive-feedback systems.}
\newblock {\em Proceedings of the National Academy of Sciences of the United
  States of America}, 101(7):1822--7, February 2004.

\bibitem{Astrom2008}
Karl~Johan Astr\"{o}m and Richard~M. Murray.
\newblock {\em {Feedback Systems: An Introduction for Scientists and
  Engineers}}.
\newblock Princeton University Press, January 2008.

\bibitem{Berg2002}
Jeremy~M Berg, John~L Tymoczko, and Lubert Stryer.
\newblock {\em {Biochemistry}}.
\newblock W.H. Freeman, New York, 2002.

\bibitem{BrownJamesWardChurchill1996}
James~Ward Brown and Ruel~V. Churchill.
\newblock {\em {Complex Variables and Applications}}.
\newblock McGraw-Hill, Inc., New York, 1996.

\bibitem{Dasgupta1988}
Soura Dasgupta.
\newblock {Kharitonov's theorem revisited}.
\newblock {\em Systems \& Control Letters}, 11(5):381--384, November 1988.

\bibitem{Dittrich2007}
Peter Dittrich and Pietro~Speroni di~Fenizio.
\newblock {Chemical organisation theory.}
\newblock {\em Bulletin of mathematical biology}, 69(4):1199--231, May 2007.

\bibitem{Eigen1978}
Manfred Eigen and Peter Schuster.
\newblock {The Hypercycle}.
\newblock {\em Naturwissenschaften}, 65(1):7--41, January 1978.

\bibitem{Elowitz2000}
M~B Elowitz and S~Leibler.
\newblock {A synthetic oscillatory network of transcriptional regulators.}
\newblock {\em Nature}, 403(6767):335--8, January 2000.

\bibitem{Feinberg1987}
Martin Feinberg.
\newblock {Chemical reaction network structure and the stability of complex
  isothermal reactors—I. The deficiency zero and deficiency one theorems}.
\newblock {\em Chemical Engineering Science}, 42(10):2229--2268, January 1987.

\bibitem{Feinberg1988}
Martin Feinberg.
\newblock {Chemical reaction network structure and the stability of complex
  isothermal reactors—II. Multiple steady states for networks of deficiency
  one}.
\newblock {\em Chemical Engineering Science}, 43(1):1--25, January 1988.

\bibitem{Fell1992}
D~A Fell.
\newblock {Metabolic control analysis: a survey of its theoretical and
  experimental development.}
\newblock {\em The Biochemical journal}, 286 ( Pt 2:313--30, September 1992.

\bibitem{Fisk2006}
Steve Fisk.
\newblock {Polynomials, roots, and interlacing}.
\newblock page 784, December 2006.

\bibitem{Fogel2006}
David~B. Fogel.
\newblock {Evolutionary Computation: Toward a New Philosophy of Machine
  Intelligence (IEEE Press Series on Computational Intelligence)}.
\newblock January 2006.

\bibitem{Fung2005}
Eileen Fung, Wilson~W Wong, Jason~K Suen, Thomas Bulter, Sun-gu Lee, and
  James~C Liao.
\newblock {A synthetic gene-metabolic oscillator.}
\newblock {\em Nature}, 435(7038):118--22, May 2005.

\bibitem{Gehrmann2011}
Eva Gehrmann, Christine Gl\"{a}\ss~er, Yaochu Jin, Bernhard Sendhoff, Barbara
  Drossel, and Kay Hamacher.
\newblock {Robustness of glycolysis in yeast to internal and external noise}.
\newblock {\em Physical Review E}, 84(2):021913--, August 2011.

\bibitem{Grimbs2007}
Sergio Grimbs, Joachim Selbig, Sascha Bulik, Hermann-Georg Holzh\"{u}tter, and
  Ralf Steuer.
\newblock {The stability and robustness of metabolic states: identifying
  stabilizing sites in metabolic networks.}
\newblock {\em Molecular systems biology}, 3:146, January 2007.

\bibitem{Gross2004}
Thilo Gross and Ulrike Feudel.
\newblock {Analytical search for bifurcation surfaces in parameter space}.
\newblock {\em Physica D: Nonlinear Phenomena}, 195(3-4):292--302, August 2004.

\bibitem{Gross2009}
Thilo Gross, Lars Rudolf, Simon~A Levin, and Ulf Dieckmann.
\newblock {Generalized models reveal stabilizing factors in food webs.}
\newblock {\em Science (New York, N.Y.)}, 325(5941):747--50, August 2009.

\bibitem{Kreyssig2012}
Peter Kreyssig, Gabi Escuela, Bryan Reynaert, Tomas Veloz, Bashar Ibrahim, and
  Peter Dittrich.
\newblock {Cycles and the qualitative evolution of chemical systems.}
\newblock {\em PloS one}, 7(10):e45772, January 2012.

\bibitem{Kuehn2013}
Christian Kuehn and Thilo Gross.
\newblock {Nonlocal generalized models of predator-prey systems}.
\newblock {\em Discrete and Continuous Dynamical Systems - Series B},
  18(3):2013, 2013.

\bibitem{Link2013}
Hannes Link, Karl Kochanowski, and Uwe Sauer.
\newblock {Systematic identification of allosteric protein-metabolite
  interactions that control enzyme activity in vivo}.
\newblock {\em Nature Biotechnology}, advance on, March 2013.

\bibitem{May1973}
R~M May.
\newblock {Stability and complexity in model ecosystems.}
\newblock {\em Monographs in population biology}, 6:1--235, January 1973.

\bibitem{Piedrafita2010}
Gabriel Piedrafita, Francisco Montero, Federico Mor\'{a}n, Mar\'{\i}a~Luz
  C\'{a}rdenas, and Athel Cornish-Bowden.
\newblock {A simple self-maintaining metabolic system: robustness,
  autocatalysis, bistability.}
\newblock {\em PLoS computational biology}, 6(8):9, January 2010.

\bibitem{Ravasz2002}
E~Ravasz, A~L Somera, D~A Mongru, Z~N Oltvai, and A~L Barab\'{a}si.
\newblock {Hierarchical organization of modularity in metabolic networks.}
\newblock {\em Science (New York, N.Y.)}, 297(5586):1551--5, August 2002.

\bibitem{Reznik2013}
Ed~Reznik, Tasso~J. Kaper, and Daniel Segrè.
\newblock {The dynamics of hybrid metabolic-genetic oscillators}.
\newblock {\em Chaos: An Interdisciplinary Journal of Nonlinear Science},
  23(1):013132, March 2013.

\bibitem{Reznik2010}
Ed~Reznik and Daniel Segr\`{e}.
\newblock {On the stability of metabolic cycles.}
\newblock {\em Journal of theoretical biology}, 266(4):536--49, October 2010.

\bibitem{Shinar2010}
Guy Shinar and Martin Feinberg.
\newblock {Structural sources of robustness in biochemical reaction networks.}
\newblock {\em Science (New York, N.Y.)}, 327(5971):1389--91, March 2010.

\bibitem{Steuer2006}
Ralf Steuer, Thilo Gross, Joachim Selbig, and Bernd Blasius.
\newblock {Structural kinetic modeling of metabolic networks.}
\newblock {\em Proceedings of the National Academy of Sciences of the United
  States of America}, 103(32):11868--73, August 2006.

\bibitem{Strogatz1994}
Steven~H. Strogatz.
\newblock {\em {Nonlinear Dynamics and Chaos}}.
\newblock Westview Press, Cambridge, 1994.

\bibitem{vanNes2009}
P~van Nes, D~Bellomo, M~J~T Reinders, and D~de~Ridder.
\newblock {Stability from structure: metabolic networks are unlike other
  biological networks.}
\newblock {\em EURASIP journal on bioinformatics \& systems biology},
  2009(1):630695, January 2009.

\end{thebibliography}

\section*{Figure Legends}

\begin{figure}[!ht]
	\centering
		\includegraphics[width=0.5\columnwidth]{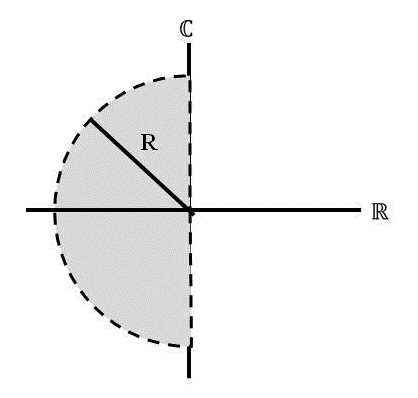}
	\label{fig:contour}
	\caption{The semicirular contour $C_R$ (dashed line) and the region it encloses (grey region) used in the proof of the cycle stability conjecture. By allowing the radius $R$ of the semicircle to go to infinity, we are able to encompass the entire left-half of the complex plane.}
\end{figure}

\begin{figure}[!ht]
	\centering
		\includegraphics[width=4in]{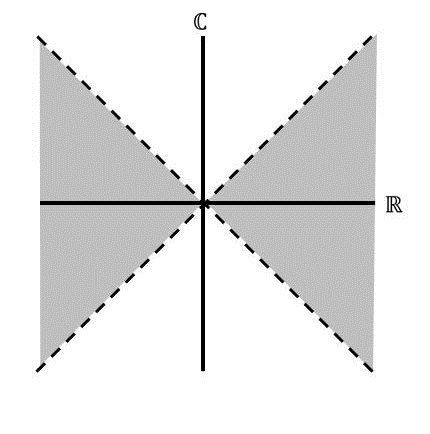}
	\label{fig:validity}
	\caption{Region of validity (in grey) for Stubborn Complex Roots Theorem. To apply the Stubborn Complex Roots Theorem, the complex roots of $P_n$ must have a smaller imaginary component than real component.}
\end{figure}

\section*{Short Title for Page Headings}
The Stubborn Roots of Metabolic Cycles

\end{document}